\documentclass{conm-p-l}

\copyrightinfo{2013}{}

\usepackage{graphicx}
\usepackage{subfigure}

\usepackage{amssymb,amsmath,amsthm,amscd}

\newtheorem{theorem}{Theorem}[section]

\theoremstyle{definition}

\theoremstyle{remark}
\newtheorem{remark}[theorem]{Remark}

\numberwithin{equation}{section}

\newcommand{\e}{\epsilon}

\newcommand{\dl}{\delta}
\newcommand{\Dl}{\Delta}

\newcommand{\ra}{\rightarrow}
\newcommand{\al}{\alpha}

\newcommand{\sg}{\sigma}

\newcommand{\pa}{\partial}

\newcommand{\na}{\nabla}

\begin{document}

\title[The distinction of turbulence from chaos]{The distinction of turbulence from chaos --- rough dependence on initial data}

\author{Y. Charles Li}
\address{Department of Mathematics, University of Missouri, 
Columbia, MO 65211, USA}
\email{liyan@missouri.edu}
\urladdr{http://www.math.missouri.edu/~cli}

%    \thanks will become a 1st page footnote.
\curraddr{}
\thanks{}

\subjclass{Primary 76, 35; Secondary 34}

\date{}

\dedicatory{}


\begin{abstract}
I propose a new theory on the nature of turbulence: when the Reynolds number is large, 
violent fully developed turbulence is due to ``rough dependence on initial data" rather 
than chaos which is caused by ``sensitive dependence on initial data"; when the Reynolds 
number is moderate, (often transient) turbulence is due to chaos. The key in the validation 
of the theory is estimating the temporal growth of the initial perturbations with the 
Reynolds number as a parameter. Analytically, this amounts to estimating the temporal 
growth of the norm of the derivative of the solution map of the Navier-Stokes equations, for
which here I obtain an upper bound $e^{C \sqrt{t Re} + C_1 t}$. This bound clearly indicates 
that when the Reynolds number is large, the temporal growth rate can potentially be large 
in short time, i.e. rough dependence on initial data. 
\end{abstract}

\maketitle

\section{Introduction}

For a long time, fluid dynamists have suspected that turbulence is ``more than" chaos.
Many chaoticians including the current author have believed that turbulence is ``no more than" 
chaos in Navier-Stokes equations. A recent result \cite{Inc12} on Euler equations forced 
the current author to have to change mind. 

The signature of chaos is ``sensitive dependence on initial data"; here I want to 
address ``rough dependence on initial data" which is very different from sensitive 
dependence on initial data. For solutions (of some system) that exhibit sensitive 
dependence on initial data, their initial small deviations usually amplify exponentially 
(with an exponent named Liapunov exponent), and it takes time for the deviations 
to accumulate to substantial amount (say order $O(1)$ relative to the small initial 
deviation). If $\e$ is the initial small deviation, and $\sg$ is the Liapunov exponent, 
then the time for the deviation to reach $1$ is about $\frac{1}{\sg} \ln \frac{1}{\e}$. On 
the other hand, for solutions that exhibit rough dependence on initial data, their initial 
small deviations can reach substantial amount instantly. Take the 3D or 2D Euler equations 
of fluids as the example, for any $t \neq 0$ (and small for local existence), the solution 
map that maps the initial condition to the solution value at time $t$ is nowhere locally 
uniformly continuous and nowhere differentiable \cite{Inc12}. In such a case, any
small deviation of the initial condition can potentially reach substantial amount 
instantly. My theory is that the high Reynolds number violent turbulence is 
due to such rough dependence on initial data, rather than sensitive dependence on 
initial data of chaos. When the Reynolds number is sufficiently large (the viscosity 
is sufficiently small), even though the solution map of the Navier-Stokes equations 
is still differentiable, but the derivative of the solution map should be potentially extremely 
large everywhere (of order $e^{C \sqrt{t Re}}$ as shown below) since the solution map of 
the Navier-Stokes equations approaches 
the solution map of the Euler equations when the viscosity approaches zero (the Reynolds 
number approaches infinity). Such everywhere large derivative of the solution map 
of the Navier-Stokes equations manifests itself as the development of violent 
turbulence in a short time. In summary, moderate Reynolds number turbulence is due to sensitive dependence on initial data of chaos, while large 
enough Reynolds number turbulence is due to rough dependence on initial 
data. This is an important new understanding on the nature of turbulence \cite{Li13}. 
One may call this the new complexity of turbulence \cite{Mit09} \cite{LLW13}.

The type of rough dependence on initial data shared by the solution map of 
the Euler equations is difficult to find in finite dimensional systems. The solution map of the Euler equations is still continuous in initial data. Such a solution 
map (continuous, but nowhere locally uniformly continuous) does not exist in finite dimensions. This may be the reason that one usually finds chaos (sensitive 
dependence on initial data) rather than rough dependence on initial data in finite dimensions. If the 
solution map of some special finite dimensional system is nowhere continuous, then 
the dependence on initial data is rough, but may be too rough to have any realistic 
application. In infinite dimensions, irregularities of solution maps are quite common, e.g. 
in water wave equations \cite{CL13} \cite{CMSW14}. 

Even though the relation between Liapunov exponent and chaos (and instability) can be complicated \cite{LK07}, generically a positive Liapunov expoent is a good indicator of chaotic dynamics. In connection with turbulence, Liapunov exponent and its extensions have been studied \cite{PV94} \cite{ABCPV97}. When the Reynolds number is moderate, homoclinic orbits, strange attractors and bifurcation routes to chaos of the Navier-Stokes equations are all dug out \cite{VK11} \cite{KUV12} \cite{KE12}. To distinguish that turbulence is exhibiting rough or sensitive dependence on initial data, one needs to study the derivative of the solution map.

\section{Derivative of the solution map}

Let $S^t$ be the solution map which maps the initial value $u(0)$ to the solution's value $u(t)$ at time $t$. So for any fixed time $t$, $S^t$ is a map defined on the phase space. The temporal growth of the norm of the derivative $DS^t$ of the solution map $S^t$ describes the amplification of the initial perturbation. The well-known Liapunov exponent is defined by $DS^t$:
\[
\sg = \lim_{t \ra +\infty} \frac{1}{t} \ln \| DS^t \| .
\]
A positive Liapunov exponent implies that nearby orbits deviate exponentially in time, i.e. sensitive dependence on initial data. The Liapunov exponent is a measure of long term temporal growth of the norm 
of the derivative $DS^t$. The temporal property of the norm of $DS^t$ can of course be much more complicated than simple long term exponential growth. In particular, the norm of $DS^t$ can be large in short time (i.e. super fast temporal growth). In such a case, the dynamics (described by $S^t$) exhibits short term unpredictability (i.e. rough dependence on initial data). One can define the following exponent
\[
\eta = \lim_{t \ra 0^+} \frac{1}{t^\al} \ln \| DS^t \|, \ \ \text{where } \al >0 . 
\]
When $\eta$ is large (e.g. approaching infinity as a parameter approaches a limit), one has short term unpredictability. In the case of Navier-Stokes equations to be studied later, $\eta$ can potentially be as large as $C \sqrt{Re}$ with $\al = 1/2$.

\section{The derivative estimate for Navier-Stokes equations} 

To verify the rough dependence on initial data for the solution map of the Navier-Stokes equations, I 
need to estimate the temporal growth of the norm of the derivative of the solution map of the Navier-Stokes equations. The Navier-Stokes equations are given by
\begin{eqnarray}
& & u_t + \frac{1}{Re} \Dl u  = - \na p - u\cdot \na u , \label{ONS1} \\
& & \na \cdot u = 0 , \label{ONS2}
\end{eqnarray}
where $u$ is the $d$-dimensional fluid velocity ($d=2,3$), $p$ is the fluid pressure, and 
$Re$ is the Reynolds number. Applying the Leray projection, one gets
\begin{equation}
u_t + \frac{1}{Re} \Dl u  = - \mathbb{P} \left ( u\cdot \na u \right ) . \label{NS} 
\end{equation}
The Leray projection is an orthogonal projection in $L^2(\mathbb{R}^d)$, given by
\[
\mathbb{P} g = g - \na \Dl^{-1} \na \cdot g . 
\]
Setting the Reynolds number to infinity $Re = 0$, the Navier-Stokes equation (\ref{NS}) reduces to 
the Euler equation
\begin{equation}
u_t = - \mathbb{P} \left ( u\cdot \na u \right ) . \label{E} 
\end{equation}
Let $H^n(\mathbb{R}^d)$ be the Sobolev space of divergence free fields. By the local wellposedness 
result of Kato \cite{Kat72} \cite{Kat75}, when $n > \frac{d}{2} +1$ ($d=2,3$), for any 
$u \in H^n(\mathbb{R}^d)$, there 
is a neighborhood $B$ and a short time $T>0$, such that for any $v \in B$ there exists a unique 
solution to the Navier-Stokes equation (\ref{NS}) in $C^0([0,T]; H^n(\mathbb{R}^d))$; as $Re 
\ra \infty$, this solution converges to that of the Euler equation (\ref{E}) in the same space. For any 
$t \in [0, T]$, let $S^t$ be the solution map:
\begin{equation}
S^t \  :  \   B  \mapsto H^n(\mathbb{R}^d), \  S^t (u(0)) = u(t), 
\end{equation}
which maps the initial condition to the solution's value at time $t$. The solution map 
is continuous for both Navier-Stokes equation (\ref{NS}) and Euler equation (\ref{E}) \cite{Kat72} 
\cite{Kat75}. A recent result of Inci \cite{Inc12} shows that for Euler equation (\ref{E}) the solution map is 
nowhere differentiable. Then it is natural to theorize that the norm of the derivative of the solution 
map approaches infinity (at most places) as the Reynolds number approaches infinity. Estimating the temporal growth of the norm of the derivative of the solution map is a daunting task. The entire subject of hydrodynamic stability is a special case where the base solution (where the derivative of the solution map is taken) is steady. Below I obtain an upper bound on the temporal growth of the norm of the derivative of the solution map. I believe the upper bound is sharp, i.e. there is no smaller upper bound.

\begin{theorem} 
The norm of the derivative of the solution map of Navier-Stokes equation (\ref{NS})  has the 
upper bound,
\begin{equation}
\| DS^t(u(0)) \| \leq e^{C \sqrt{t Re} + C_1 t}, \label{UB} 
\end{equation}
where
\[
C = \frac{8}{\sqrt{2e}} \max_{\tau \in [0,T]} \| u(\tau )\|_n, \ \ C_1 = 4\max_{\tau \in [0,T]} \| u(\tau )\|_n
= \frac{\sqrt{2e}}{2} C.
\]
\label{UBT}
\end{theorem}
\begin{proof}
Applying the method of variation of parameters, one converts the Navier-Stokes equation (\ref{NS}) into
the integral equation
\begin{equation}
u(t) = e^{\frac{t}{Re}\Dl } u(0) - \int_0^t e^{\frac{t-\tau }{Re}\Dl } \mathbb{P} \left ( u\cdot \na u \right ) d\tau .                   \label{INS} 
\end{equation}
Taking the differential in $u(0)$, one gets the differential form
\begin{equation}
du(t) = e^{\frac{t}{Re}\Dl } du(0) - \int_0^t e^{\frac{t-\tau }{Re}\Dl } \mathbb{P} \left ( du\cdot \na u 
+u \cdot \na du \right ) d\tau .   \label{DNS} 
\end{equation}
The norm of the derivative $DS^t(u(0)) = \pa u(t) / \pa u(0)$ is given by 
\[
\| DS^t(u(0)) \| = \sup_{du(0)} \frac{\| du(t)\|_n}{\| du(0)\|_n} .
\]
Applying the inequality 
\[
\| e^{\frac{t}{Re}\Dl } u \|_n \leq \left (\frac{1}{\sqrt{2e}} \sqrt{\frac{Re}{t}} + 1 \right )\| u \|_{n-1} ,
\]
one gets
\begin{eqnarray*}
& & \| du(t) \|_n \leq  \| du(0) \|_n + \\
& & 4 \max_{\tau \in [0,T]} \| u(\tau )\|_n  \int_0^t 
\left ( \frac{\sqrt{Re}}{\sqrt{2e}} \frac{1}{\sqrt{t-\tau }} +1 \right )\| du(\tau ) \|_n  d\tau .
\end{eqnarray*}
Applying the Gronwall's inequality, one gets the estimate (\ref{UB}).
\end{proof}

\begin{remark}
By Theorem \ref{UBT}, for any initial perturbation $\dl u(0)$, the deviation of the corresponding 
solutions can potentially amplifies according to
\[
\| \dl u(t) \|_n \leq  e^{C \sqrt{t Re}\ + \ C_1 t} \| \dl u(0) \|_n .
\]
When the Reynolds number is large, the amplification can
potentially reach substantial amount in short time.  
\end{remark} 

\begin{remark}
The same upper bound (\ref{UB}) also holds for the periodic boundary condition, i.e. when the Navier-Stokes equations (\ref{ONS1})-(\ref{ONS2}) are defined on $d$-dimensional torus $\mathbb{T}^d$ instead of the whole space $\mathbb{R}^d$. 
\end{remark} 

\begin{remark}
The beauty of the upper bound (\ref{UB}) can be revealed when the base solution (where the derivative of the solution map is taken) is steady. In such a case, one is dealing with hydrodynamic stability theory. The zero-viscosity limit of the eigenvalues of the linear Navier-Stokes equations at the steady state can be complicated \cite{LL08}. In the zero-viscosity limit, some of the eigenvalues may persist to be the eigenvalues of the corresponding linear Euler equations \cite{Li05}; some eigenvalues may condense into continuous spectra; and other eigenvalues may approach a set that is not in the spectra of the corresponding linear Euler equations. The $C_1t$ exponent in (\ref{UB}) covers the growth induced by persistent unstable eigenvalues, while the $C\sqrt{tRe}$ exponent in (\ref{UB}) covers the growth induced by the rest eigenvalues. When $Re$ is large, the $C\sqrt{tRe}$ can be large in short time. During such short time, stable eigenvalues do not imply ``decay". Even though its derivative does not exist, directional derivatives of the solution map of Euler equations can exist as shown by the existence of solutions to the well-known Rayleigh equations. The unbounded continuous spectrum \cite{Li05} of the linear Euler equations leads to the nonexistence of the derivatives of the solution map of Euler equations.
\end{remark} 

\begin{remark}
The upper bound (\ref{UB}) is sharp when the base solution (where the derivative of the solution map is taken) is the zero solution $u(t)=0$. In this case, 
\[
\| DS^t(0)\| = 1 ,
\]
and the upper bound (\ref{UB}) is also $1$. In general, estimating the lower bound of $\| DS^t(u(0))\|$ may only be done on a case by case base for the base solutions. When the base solutions are steady, this is the theory of hydrodynamic instability.
\end{remark}

\end{document}